\documentclass[12pt]{amsart}
\usepackage[all]{xy}
\usepackage{amssymb,amsmath,amsfonts}
\usepackage{enumerate}
\usepackage{comment}
\newtheorem{theorem}{Theorem}
\newtheorem{lemma}[theorem]{Lemma}
\newtheorem{proposition}[theorem]{Proposition}

\theoremstyle{definition}

\theoremstyle{remark}
\newtheorem{remark}[theorem]{Remark}

\numberwithin{equation}{section}
\numberwithin{theorem}{section}

\def\CC{\mathbb{C}}
\def\QQ{\mathbb{Q}}

\def\ZZ{\mathbb{Z}}

\DeclareMathOperator\End{End}

\DeclareMathOperator\Hom{Hom}

\DeclareMathOperator\Ind{Ind}

\def\vac{{\boldsymbol{1}}}  

\def\ii{\mathrm{i}} 
\def\al{\alpha}
\def\be{\beta}
\def\ga{\gamma}
\def\de{\delta}

\def\ep{\varepsilon}

\def\la{\lambda}
\def\si{\sigma}

\def\<{\left\langle}
\def\>{\right\rangle}

\def\lieh{{\mathfrak{h}}}

\def\F{\mathcal{F}}

\begin{document}
\title[
Quantum dimensions and fusion products ]{Quantum dimensions and fusion products for irreducible $V_Q^\si$-modules with $\si^2=1$}

\author{Jason Elsinger}
\address{Department of Mathematics\\
Spring Hill College\\
Mobile, AL 36695, USA}
\email{jelsinger@shc.edu}


\date{\today}



\begin{abstract}
Every isometry $\sigma$ of a positive-definite even lattice $Q$ can be lifted to an automorphism of 
the lattice vertex algebra $V_Q$. An important problem in vertex algebra theory and conformal field theory is to classify the
representations of the $\sigma$-invariant subalgebra $V_Q^\sigma$ of $V_Q$,
known as an orbifold. In the case when $\sigma$ is an isometry of $Q$ of order two, we have classified the irreducible modules of the
orbifold vertex algebra $V_Q^\sigma$ and identified them as submodules of twisted or untwisted $V_Q$-modules in \cite{BE}. Here we calculate their quantum dimensions and fusion products. 

The examples where $Q$ is the orthogonal direct sum of two copies of the $A_2$ root lattice and $\sigma$ is the 2-cycle permutation as well as where $Q$ is the $A_n$ root latice and $\si$ is a Dynkin diagram automorphism are presented in detail.
\end{abstract}

\maketitle


\section{Introduction}\label{sintro}

The notion of a vertex algebra 
has been a powerful tool for studying representations of infinite-dimensional Lie algebras.
The theory of vertex algebras was introduced by Borcherds \cite{B} and has since been developed in a number of works (see e.g. \cite{FLM, K2, FB, LL, KRR}).


Let $\si$ be an automorphism of a vertex algebra $V$. An {\it orbifold} is the subalgebra of $\si$-invariants in $V$, denoted $V^\si$ (see \cite{DVVV, KT, DLM2}). It will follow from definitions below that every $\si$-twisted representation of $V$ becomes untwisted when
restricted to $V^\si$. 

It has been a long-standing conjecture that all irreducible $V^\si$-modules are obtained by restriction from twisted or untwisted $V$-modules. This conjecture has recently been proved in a series of works by M. Miyamoto (see \cite{M1,M2,M3, MC}) under suitable assumptions. In particular, it is shown in \cite{MC} that
if $V$ is a simple regular vertex operator algebra of CFT type, then the fixed point subalgebra $V^\si$ for a finite automorphism $\si$ of $V$ is also regular.

In this paper, we are concerned with the case when $Q$ is a positive-definite even lattice and $\si$ is an isometry of $Q$ of order two. 
In \cite{BE}, we have classified and explicitly constructed all irreducible modules of the orbifold vertex algebra $V_Q^\si$, and have realized them as submodules of twisted or untwisted $V_Q$-modules. This paper is a continuation, where we give explicitly the quantum dimensions and fusion products for the irreducible $V_Q^\si$-modules.

Our general approach is to restrict from $V_Q^\si$ to $V_L^\si$, where $L$ is the sublattice of $Q$ spanned by eigenvectors of $\si$. The subalgebra $V_L^\si$ \, then factors as a tensor product $V_{L_+}\otimes V_{L_-}^\si$, where 
\[L_\pm=\{\al\in Q\,|\,\si\al=\pm\al\},\]
and $V_L^\si$ is rational by results of \cite{DLM3,ABD, A2, DJL, FHL}. 
We then use the known irreducible $V_L^\si$-modules (see \cite{DN,AD}) and their intertwining operators (see \cite{A1,ADL}) to determine the quantum dimensions and fusion products for the irreducible $V_Q^\si$-modules given in \cite{BE}.

The paper is organized as follows. In Section \ref{sva}, we briefly review lattice vertex algebras, their twisted modules, and the results for $\si=-1$ that we need. We also review the notions of quantum dimension and fusion product including results from \cite{DJX} that we need. In Section \ref{irrep}, we review the irreducible orbifold modules given in \cite{BE}.
Our main results are presented in Sections \ref{qd1} and \ref{fp1}. Examples where $Q$ is the orthogonal direct sum of two copies of the $A_2$ root lattice and $\sigma$ is the 2-cycle permutation as well as where $Q$ is the $A_n$ root latice and $\si$ is a Dynkin diagram automorphism are presented in detail
in Section \ref{ex}.

\section{Lattice Vertex algebras and their twisted modules}\label{sva}

In this section, we briefly review lattice vertex algebras and their twisted modules, and
recall the results in the case when $\si=-1$. These results will aid in the general description.
General references on vertex algebras are e.g. \cite{FLM, K2, FB, LL, KRR}, among many other works. 

\subsection{Vertex algebras and twisted modules}\label{vert}

A \emph{vertex algebra} 
is a vector space $V$ with a {\it vacuum vector} \,$\vac\in V$ 
together with a linear map 
\begin{equation}\label{vert2}
Y(\cdot,z)\cdot \colon V \otimes V \to  V[[z]][z^{-1}] \,
\end{equation}
satisfying the axioms listed below.
For $v\in V$, we have 
\begin{equation}
Y(v,z)=\displaystyle\sum_{n\in\ZZ}v_{(n)}z^{-n-1},\qquad v_{(n)}\in \End V.
\end{equation}
The vacuum vector is the identity in the sense that
\begin{equation*}
a_{(-1)}\vac = \vac_{(-1)} a = a \,, \qquad 
a_{(n)}\vac = 0 \,, \quad n\geq 0 \,.
\end{equation*}
In particular, 
\[
\displaystyle\lim_{z\to0}Y(a,z)\vac=a.
\]
For every $a\in V$, 
we call the image $Y(a,z)$ a {\it field}, meaning it can be viewed as
a formal power series from $(\End V)[[z,z^{-1}]]$
involving only finitely many negative powers of $z$ when
applied to any vector.

The main axiom for a vertex algebra is the \emph{Borcherds identity}
(also called Jacobi identity in \cite{FLM})
satisfied by the fields:
\begin{equation}\label{vert5}
\begin{split}
y^{-1}\de\left(\frac{z_1-z_2}{y}\right)&Y(u,z_1)Y(v,z_2)-y^{-1}\de\left(\frac{z_2-z_1}{-y}\right)Y(v,z_2)Y(u,z_1)\\
&=z_2^{-1}\de\left(\frac{z_1-y}{z_2}\right)Y(Y(u,y)v,z_2),
\end{split}
\end{equation}
where $u,v \in V$ and $\de$ is the delta function $\de(z)=\sum_{n\in\ZZ}z^n$. 
It is important to note that the commutator $[Y(u,z_1),Y(v,z_2)]$ follows from Borcherds identity be taking residue of both sides with respect to $y$.

A  $V$-\emph{module} for a vertex algebra $V$ is a vector space $M$ endowed with a
linear map $Y^M(\cdot,z)\cdot \colon V \otimes M \to M[[z]][z^{-1}]$
(cf.\ \eqref{vert2}) such that the Borcherds identity
\eqref{vert5} holds for all $u,v\in V$ which act as operators on $M$ (see \cite{FB, LL, KRR}).
Let $\si$ be an automorphism of $V$ of finite order $r$. Then $\si$ is diagonalizable.
In the definition of a \emph{$\si$-twisted \,$V$-module} $M$ \cite{FFR, D2}, the image of $Y^M$ is allowed to have nonintegral rational powers of $z$:
\begin{equation*}
Y(u,z) = \sum_{n\in p+\ZZ} u_{(n)} \, z^{-n-1} \,, \qquad
\text{if} \quad \si u = e^{-2\pi\ii p} u \,, \; p\in\frac1r\ZZ \,,
\end{equation*}
where $u_{(n)} \in \End M$.
The Borcherds identity in the twisted case is similar to \eqref{vert5}, and it requires that $u$ be an eigenvector of $\si$ (see \cite{FLM, FFR, D2}).

%

Recall from \cite{FHL} that if $V_1$ and $V_2$ are vertex algebras, their tensor product is again a vertex algebra via 
\begin{equation*}
Y(v_1\otimes v_2,z) = Y(v_1,z) \otimes Y(v_2,z) \,, \qquad v_i \in V_i \,.
\end{equation*}
Furthermore, if $M_i$ is a $V_i$-module, then the above formula defines the structure of a $(V_1\otimes V_2)$-module
on $M_1\otimes M_2$ (see \cite{FHL}). A similar statement is also true for twisted modules (see \cite[Lemma 2.2]{BE}). 


\subsection{Lattice vertex algebras}\label{lat}

Consider an \emph{integral lattice}, i.e., a free abelian group $Q$ of finite rank together with a symmetric
nondegenerate bilinear form $(\cdot|\cdot) \colon Q\times Q\to\ZZ$. We assume that $Q$ is \emph{even},
i.e., $|\al|^2=(\al|\al) \in2\ZZ$ for all $\al\in Q$.
The corresponding complex 
vector space
\,$\lieh = \CC\otimes_\ZZ Q$\, is  considered as an abelian Lie algebra with a bilinear form extended from \,$Q$.

The \emph{Heisenberg algebra}
$\hat\lieh = \lieh[t,t^{-1}] \oplus \CC K$
is the Lie algebra with brackets $[K,\lieh]=0$, and
\begin{equation}\label{heis1}
[a_m,b_n] = m \delta_{m,-n} (a|b) K \,, \qquad
a_m=at^m \,.
\end{equation}
Its induced irreducible highest-weight representation 
\begin{equation*}
M(1)= \Ind^{\hat\lieh}_{\lieh[t]\oplus\CC K} \CC \cong S(\lieh[t^{-1}]t^{-1})
\end{equation*}
on which $K=1$ is known as the (bosonic) \emph{Fock space}.

Following \cite{FK,B}, let $\ep\colon Q \times Q \to \{\pm1\}$ be a bimultiplicative $2$-cocycle
such that
\begin{equation}\label{lat2}
\ep(\al,\al) = (-1)^{|\al|^2/2}  \,,
\qquad \al\in Q \,.
\end{equation}
Then we have the associative algebra $\CC_\ep[Q]$ with basis
$\{ e^\al \}_{\al\in Q}$ and multiplication
\begin{equation}\label{lat1}
e^\al e^\be = \ep(\al,\be) e^{\al+\be} \,.
\end{equation}
Such a $2$-cocycle $\ep$ is unique up to equivalence.
In addition,
\begin{equation}\label{lat22}
\ep(\al,\be) \ep(\be,\al) = (-1)^{(\al|\be)}  \,, 
\qquad \al,\be\in Q \,.
\end{equation}

The \emph{lattice vertex algebra} 
associated to $Q$ is defined as $V_Q=M(1)\otimes\CC_\ep[Q]$,
where the vacuum vector is $\vac=1\otimes e^0$.
An action of the Heisenberg algebra on $V_Q$ is given by
\begin{equation*}
a_n e^\be = \delta_{n,0} (a|\be) e^\be \,, \quad n\geq0 \,, \qquad
a\in\lieh \,.
\end{equation*}
%
The map $Y$ on $V_Q$ is uniquely determined by the
generating fields:
\begin{align}\label{lat4}
Y(a_{-1}\vac,z) &= \sum_{n\in\ZZ} a_n \, z^{-n-1} \,, \qquad a\in\lieh \,,
\\ \label{lat5}
Y(e^\al,z) &= e^\al z^{\al_0} 
\exp\Bigl( \sum_{n<0} \al_n \frac{z^{-n}}{-n} \Bigr) 
\exp\Bigl( \sum_{n>0} \al_n \frac{z^{-n}}{-n} \Bigr) \,,
\end{align}
where $z^{\al_0} e^\be = z^{(\al|\be)} e^\be$.
Since the map $\lieh\to M(1)$ given by $a\mapsto a_{-1}\vac$ is injective,
we can identify
$a\in\lieh$ with $a_{-1}\vac \in M(1)$, i.e., $a_{(n)}=a_n$ for all $n\in\ZZ$. 

\subsection{Twisted representations of lattice vertex algebras}\label{twlat}
Suppose $\si$ is an automorphism of $\lieh$ with finite order $r$ which preserves the bilinear form $(\cdot|\cdot)$. Then $\si$ is naturally extended to both $\hat\lieh$ and $M(1)$, again denoted $\si$.
The \emph{$\si$-twisted Heisenberg algebra} 
$\hat\lieh_\si$ 
is spanned over $\CC$ by $K$ and
the elements $a_m = at^m$ such that $\si a = e^{-2\pi\ii m} a$. It becomes a Lie algebra with bracket 
\eqref{heis1} for $m,n\in \frac1r\ZZ$.

Let $\hat\lieh_\si^\ge$ (respectively, $\hat\lieh_\si^<$) be the abelian subalgebra of 
$\hat\lieh_\si$ spanned by all elements $a_m$ with $m\geq0$ 
(respectively, $m<0$). We let $\hat\lieh_\si^\ge$ act on $\CC$ trivially and $K$ act as the identity operator.
The \emph{$\si$-twisted Fock space} is defined as the induced module
\begin{equation}\label{twheis2}
M(1)_\si = \Ind^{\hat\lieh_\si}_{\hat\lieh_\si^\ge \oplus\CC K} \CC \cong S(\hat\lieh_\si^<) \,,
\end{equation}
and is an irreducible highest-weight representation of $\hat\lieh_\si$ which
has the structure of a $\si$-twisted representation of the vertex algebra $M(1)$
(see \cite{FLM,KRR}). 
For the map \,$Y$, we let $Y(\vac,z)$ be the identity operator,
\begin{equation*}
Y(a,z) = \sum_{n\in p+\ZZ} a_{n} \, z^{-n-1} \,, \qquad
a\in\lieh \,, \;\; \si a = e^{-2\pi\ii p} a \,, \;\;  p\in\frac1r\ZZ \,,
\end{equation*}
and extend linearly to all of $\lieh$.

Now consider $\si$ as an automorphism of $Q$. Since $\si$ preserves the bilinear form, the uniqueness of the cocycle $\ep$ and \eqref{lat22} imply that
\begin{equation}\label{twlat3}
\eta(\al+\be) \ep(\si\al,\si\be) = \eta(\al)\eta(\be) \ep(\al,\be)
\end{equation}
for some function $\eta\colon Q\to\{\pm1\}$. It is shown in \cite[Lemma 2.3]{BE} that we can set $\eta=1$ on any sublattice whose elements satisfy \/ $\ep(\si\al,\si\be) = \ep(\al,\be)$.
In particular, $\eta$ can be chosen such that
\begin{equation}\label{twlat2}
\eta(\al)=1 \,,\qquad \al\in Q\cap\lieh^\si \,,
\end{equation}
where $\lieh^\si\subset\lieh$ denotes the subspace spanned by vectors fixed under $\si$.
Then $\si$ can be lifted to an automorphism of 
the lattice vertex algebra $V_Q$ by setting
\begin{equation}\label{twlat4}
\si(a_n)=\si(a)_n \,, \quad \si(e^\al)=\eta(\al) e^{\si\al} \,,
\qquad a\in\lieh \,, \; \al\in Q \,.
\end{equation}
Note the order of $\si$ on $V_Q$ may double.

It is shown in \cite{BK} that irreducible $\si$-twisted $V_Q$-modules are parameterized by the set \,$(Q^*/Q)^\si$ of $\si$-invariants in $Q^*/Q$ and every $\si$-twisted $V_Q$-module is a direct sum of irreducible ones (see \cite[Theorem 4.2]{BK}). 
%
In the special case when $\si=1$, we get Dong's Theorem that the irreducible $V_Q$-modules
are classified by $Q^*/Q$ (see \cite{D1}). Explicitly, they are given by:
\begin{equation*}
V_{\la+Q} = \F\otimes\CC_\ep[Q] e^\la\,, \qquad \la\in Q^* \,.
\end{equation*}

When the lattice $Q$ is written as an orthogonal direct sum of sublattices, $Q=L_1\oplus L_2$, we have a natural isomorphism
$V_Q \cong V_{L_1} \otimes V_{L_2}$. 
It can be shown that if $L_1$ and $L_2$ are $\si$-invariant, there is
a correspondence of irreducible twisted modules in the sense that every irreducible\/ $\si$-twisted\/ $V_Q$-module\/ $M$ is a tensor product, $M\cong M_1\otimes M_2$, where\/
$M_i$ is an irreducible\/ $\si|_{L_i}$-twisted\/ $V_{L_i}$-module (cf.\ \cite{FHL, BE}).


\subsection{The case $\si=-1$}\label{tcasesi-1}

Now we review what is known in the case when $\si=-1$, which will be used in our treatment of the
general case. In this subsection, we denote the even integral lattice by $L$ instead of $Q$ so that 
$\lieh=\CC\otimes_\ZZ L$ is the corresponding complex vector space.

Let \,$\hat{L}$\, be the central extension of \,$L$\, by the 2-group \,$\langle-1\rangle$:
\begin{equation}\label{hatL}
1\to\langle-1\rangle\to\hat{L}\;\bar\to\; L\to1
\end{equation}
with the group structure in $\hat{L}$ written multiplicatively. 
Then $\epsilon$ is the associated 2-cocycle (cf. \eqref{lat1}). Let 
\begin{equation}\label{K}
K=\{\si(a)a^{-1}\,|\,a\in\hat{L}\}
\end{equation}
and \,$T$\, be any $\hat{L}/K$-module with the natural action of $-1$. Then these irreducible $\hat{L}/K$-modules are determined by central characters $\chi$ of \,$\hat{L}/K$\, such that
$\chi(-1)=-1$. When convenient, we denote the corresponding module by \,$T_\chi$.
Define the vector space $V_L^T=M(1)_\si\otimes T$ (cf.\ \eqref{twheis2}).
Then every irreducible $\si$-twisted $V_L$-module is isomorphic to $V_L^T$ for some irreducible $\hat{L}/K$-module $T_\chi$
(see \cite{FLM, D2}). 

An action of $\si$ on $V_L^T$ can be defined by
\begin{equation}\label{siact2}
\si\left(h^1_{(-n_1)}\cdots h^k_{(-n_k)}t\right)=(-1)^kh^1_{(-n_1)}\cdots h^k_{(-n_k)}t
\end{equation}
for $h^i\in\lieh$, $n_i\in\frac{1}{2}+\ZZ_{\geq 0}$ and $t\in T$. Then $\si$ is an automorphism of $V_L^\si$-modules. The eigenspaces for $\si$ are denoted $V_L^{T,\pm}$ and are both $V_L^\si$-modules. Similarly, an action of $\si$ on the untwisted $V_L$-modules $V_{\la+L}$ can be defined by
\begin{equation}\label{siact1}
\si\left(h^1_{(-n_1)}\cdots h^k_{(-n_k)}e^{\la+\al}\right)=(-1)^kh^1_{(-n_1)}\cdots h^k_{(-n_k)}e^{-\la-\al}
\end{equation}
for $h^i\in\lieh$, $n_i\in\ZZ_{\geq 0}$ and $\al\in L$.
Clearly $\si V_{\la+L} \subseteq V_{-\la+L}$ which implies the eigenspaces $V_{\la+L}^\pm$ are $V_L^\si$-modules for $\la\in L^*$ with $2\la\in L$ and $V_{\la+L} \cong V_{-\la+L}$ as $V_L^\si$-modules if $2\la\not\in L$. The following theorem is the classification of irreducible $V_L^\si$-modules.

\begin{theorem}[\cite{DN, AD}]\label{AD}
Let\/ $L$ be a positive-definite even lattice and\/ $\si=-1$ on\/ $L$.
Then any irreducible admissible\/ $V_L^\si$-module is isomorphic to one of the following$:$
\[
V_{\la+L}^\pm\quad (\la\in L^*, \,2\la\in L),\quad V_{\la+L}\quad (\la\in L^*, \,2\la\notin L),\quad V_L^{T,\pm},
\]
where\/ $T$ is an irreducible\/ $\hat{L}/K$-module.
\end{theorem}

Now we describe intertwining operators between the irreducible $V_L^\si$-modules. 
For a vector space $U$, denote by
\begin{equation*}
U\{z\}=\Bigl\{\sum_{n\in\QQ}v_{(n)}z^{-n-1}\Big|v_{(n)}\in U\Bigr\}
\end{equation*}
the space of $U$-valued formal series involving rational powers of $z$.
Let $V$ be a vertex algebra, and $M_1, M_2, M_3$ be $V$-modules, which are not necessarily distinct. 
Recall from \cite{FHL} that an \emph{intertwining operator} of type
$\displaystyle\binom{M_3}{M_1\;M_2}$
is a linear map given by
\begin{align*}
\mathcal{Y}\colon M_1&\to\Hom(M_2,M_3)\{z\} \,,\\
v&\mapsto \mathcal{Y}(v,z)=\sum_{n\in\QQ}v_{(n)}z^{-n-1} \,,\quad v_{(n)}\in\Hom(M_2,M_3)
\end{align*}
which satisfies $v_{(n)}u=0$, for $n\gg 0$, and the Borcherds identity \eqref{vert5} for $u\in V, v\in M_1$, where the action of each term is on $M_2$.
The intertwining operators of type
$\displaystyle\binom{M_k}{M_i\;M_j}$
form a vector space denoted $\mathcal{V}_{i\,j}^{k}$ 
and the {\it fusion rule} associated with an algebra $V$ and its modules $M_i, M_j, M_k$ is 
$N_{i,\,j}^{k}=\dim\mathcal{V}_{i\,j}^{k}$. 

An {\it admissible} $\si$-twisted $V$-module $M$ is a $\frac{1}{r}\ZZ$-graded $V$-module 
\begin{equation}\label{admissible}
M=\bigoplus_{n\in\frac{1}{r}\ZZ}M(n)
\end{equation}
for which $(v_mM)(n)\subset M(\text{wt}v-m-1+n)$ for homogeneous $v\in V$ and $m,n\in\frac{1}{r}\ZZ$.
The {\it contragredient} $M'$ is defined as the graded dual space
\begin{align}
M'=\displaystyle\bigoplus_{n\in\frac{1}{r}\ZZ}M(n)^*,
\end{align}
where $M(n)^*=\Hom_\CC(M(n),\CC)$. A $V$-module $M$ is {\it self-dual} if $M=M'$. It is shown in \cite{FHL} that $M'$ has a vertex algebra module structure and the module $M$ is irreducible if and only if $M'$ is irreducible. It is also shown in \cite{FHL} that the fusion rules satisfy a few symmetry properties:
\begin{proposition}[\cite{FHL}]
For \,$V$-modules \,$W_1$, $W_2$\, and \,$W_3$, we have
\begin{equation}
N_{1, 2}^3=N_{2, 1}^3,\quad N_{1, 2}^3=N_{1, 3'}^{2'},
\end{equation}
where prime denotes contragredient module.
\end{proposition}

The fusion rules for $V_L^\si$ and its irreducible modules
were calculated in \cite{A1, ADL} to be either zero or one. 
In order to present their theorem, we first introduce some additional notation.
For $\la\in L^*$ such that $2\la\in L$, let
\begin{align}
\label{pi}
\pi_{\la,\mu} 
&=(-1)^{|\la|^2|\mu|^2} \,, \qquad \la,\mu\in L^* \,,\\
\label{c}
c_\chi(\la)&=(-1)^{(\la|2\la)}\ep(\la,2\la)\chi(e^{2\la}) \,.
\end{align}
For any central character $\chi$ of $\hat{L}/K$, also set $\chi^{(\la)}$ and $\chi'$ to be the central characters defined by 
\begin{align}\label{newchi}
\chi^{(\la)}(a)&=(-1)^{(\la|\bar{a})}\chi(a) \,,\\
\chi'(a)&=(-1)^{\frac{1}{2}(\bar{a}|\bar{a})}\chi(a),\label{chiprime}
\end{align}
for \,$a$ in the center of $\hat{L}/K$, and set $T^{(\la)}_\chi=T_{\chi^{(\la)}}$. Note that \,$\chi'=\chi$\, whenever $4$ divides \,$(\bar{a}|\bar{a})$\, for all $a\in \hat{L}$. It is shown in \cite{ADL} that 
\begin{align}
(V_{\la+L})'&\cong V_{\la+L},\quad 2\la\notin L\\
(V_{\la+L}^\pm)'&\cong
\begin{cases}
V_{\la+L}^\pm,&2|\la|^2\in2\ZZ\\
V_{\la+L}^\mp,&2|\la|^2\in2\ZZ+1
\end{cases},\quad 2\la\in\ L\label{V'}\\
(V_L^{T_\chi,\pm})'&\cong V_L^{T_{\chi'},\pm}.
\end{align}

The following theorem is only a part of Theorem 5.1 from \cite{ADL}.

\begin{theorem}[\cite{ADL}]\label{ADL}
Let\/ $L$ be a positive-definite even lattice, and \/ $\epsilon\in\{\pm\}$. 
Then for two irreducible\/ $V_L^\si$-modules\/ $M_2,M_3$, we have the following fusion rules:
\begin{enumerate}
\item If \/ $\la\in L^*\cap\frac{1}{2}L$, the fusion rule of type\/ $\displaystyle\binom{M_3}{V_{\la+L}^{\epsilon}\;M_2}$
is equal to\/ $1$ if and only if the pair\/ $(M_2,M_3)$ is one of the following$:$
\begin{align*}
&(V_{\mu+L},V_{\la+\mu+L}), \quad \mu\in L^*, \;\; 2\mu\not\in L,\\
&(V_{\mu+L}^{\epsilon_1},V_{\la+\mu+L}^{\epsilon_2}),\quad \mu\in L^*, \;\; 2\mu\in L, \;\;\epsilon_1\in\{\pm\}, \;\;
\epsilon_2=\epsilon_1\epsilon\pi_{\la,2\mu}, \\
&(V_L^{T_\chi, \,\epsilon_1}, V_L^{T^{(\la)}_{\chi}, \,\epsilon_2}),\quad \epsilon_1\in\{\pm\}, \;\; \epsilon_2=c_\chi(\la)\epsilon_1\epsilon.
\end{align*}
\item The fusion rule of type\/ $\displaystyle\binom{M_3}{V_{L}^{T_\chi,\epsilon}\;M_2}$
is equal to\/ $1$ if and only if the pair\/ $(M_2,M_3)$ is one of the following$:$
\begin{align*}
&((V_{L}^{T_\chi^{(\la)},\pm})',V_{\la+L}), \quad (V_{\la+L}, V_{L}^{T_\chi^{(\la)},\pm}),\quad  \la\notin L^*\cap\frac{1}{2}L,\\
&((V_{L}^{T_\chi^{(\la)},\epsilon_1})',(V_{\la+L}^{\epsilon_2})'),\quad (V_{\la+L}^{\epsilon_2}, V_{L}^{T_\chi^{(\la)},\epsilon_1}), \quad  \la\in L^*\cap\frac{1}{2}L,
\end{align*}
\end{enumerate}
where \,$\epsilon_1\in\{\pm\},$  and\,
$\epsilon_2=c_\chi(\la)\epsilon_1\epsilon.$
In all other cases, the fusion rules of types\/ $\displaystyle\binom{M_3}{V_{\la+L}^{\epsilon}\;M_2}$ and $\displaystyle\binom{M_3}{V_{L}^{T_\chi,\epsilon}\;M_2}$ are zero.
\end{theorem}

\subsection{Quantum dimension and fusion product}\label{qdfp}
The notions and properties of quantum dimensions have been systematically studied in \cite{DJX}. A vertex algebra is {\it rational} if the admissible module category is semisimple (cf. \eqref{admissible}). Let $C_2(V)\subset V$ denote the subalgebra spanned by all $-2$-products. A vertex algebra is {\it $C_2$-cofinite} if \,$V/C_2(V)$\, is finite dimensional.
Suppose $V=\bigoplus_{n\in\ZZ}V_n$ is a self dual vertex operator algebra of {\it CFT type}, i.e., $V_n=0$ for $n<0$, and $V_0=\CC\vac$. 
In the case when $V$ is also rational and $C_2$-cofinite, 
quantum dimensions of its irreducible modules have nice properties and turn out to be helpful in determining their fusion products.

For a vertex algebra $V$ and a $\frac{1}{r}\ZZ$-graded $\si$-twisted $V$-module $M$, the {\it formal character} of $M$ is defined as
\begin{equation}
\text{ch}_qM=q^{\la-c/24}\sum_{n\in \frac{1}{r}\ZZ}(\dim M_{\la+n})q^n,
\end{equation}
where $\la$ is the conformal weight of $M$. The {\it quantum dimension} of $M$ over $V$ is defined as
\begin{equation}
q\dim_VM=\lim_{q\to1^-}\frac{\text{ch}_qM}{\text{ch}_qV}.
\end{equation}

Next we describe the fusion product of two irreducible $V$-modules, \,$W_1$\, and \,$W_2$. Let $I\in\mathcal{V}_{W_1\,W_2}^{W}$ be an intertwining operator. A module \,$(W,I)$ is a {\it fusion product} of \,$W_1$\, and \,$W_2$\, if for any $V$-module \,$M$\, and\, $\mathcal{Y}\in\mathcal{V}_{W_1\,W_2}^{M}$, there is a unique $V$-module homomorphism\, $\phi:W\longrightarrow M$ \,such that\, $\mathcal{Y}=\phi\,\circ I$. The fusion product $(W,I)$ is typically denoted by \,$W_1\boxtimes W_2$ and in the case when $V$ is rational, it is known that the fusion product exists:
\begin{equation}
W_1\boxtimes W_2=\sum_WN_{W_1,\,W_2}^{W}\,W,
\end{equation}
where the sum runs over all equivalence classes of irreducible $V$-modules. 
If \,$V$\, is a simple vertex algebra, a simple \,$V$-module \,$M$\, is a {\it simple current} if for any irreducible \,$V$-module \,$W$, the fusion product \,$W\boxtimes M$\, exists and is also a simple \,$V$-module. 

Now assume $V$ is a rational, $C_2$-cofinite, self dual vertex algebra of CFT type with finitely many irreducible admissible modules. Following \cite{DJX}, set $M_0,\ldots,M_d$ as the inequivalent irreducible $V$-modules with $M_0\cong V$. Under these conditions, the following properties of quantum dimension are shown in \cite{DJX}.
\begin{proposition}[\cite{DJX}]\label{DJX}
\begin{enumerate}
\item $q\dim_V M_i\geq1$ for any \,$0\leq i\leq d$.
\item If $q\dim_V M$ exists, then $q\dim_VM=q\dim_VM'.$
\item For any $0\leq i,j\leq d$,
\begin{equation*}\label{qdbox}
q\dim_V(M_i\boxtimes M_j)=q\dim_V M_i\cdot q\dim_V M_j.
\end{equation*}
\item $M$ is a simple current of $V$ if and only if $q\dim_VM=1$.
\end{enumerate}
\end{proposition}
We also list a few other basic properties of quantum dimension needed in the calculations.

\begin{lemma}
Let $A$, $A_i$, and $B$ be $V$-modules for which $A\subset B$ and $i$ is in some index set. Then 
\begin{align}
q\dim_VA&\leq q\dim_VB,\\
q\dim_V\bigoplus_i A_i&=\sum_iq\dim_VA_i.\label{sum}
\end{align}
\end{lemma}

\begin{lemma}
For vertex algebras \,$V$\, and \,$W$, let $M$ be a $V$-module and $N$ be a $W$-module. Then 
\begin{align}
q\dim_{V\otimes W}M\otimes N&=q\dim_VM\cdot q\dim_WN\label{tensor}.
\end{align}
\end{lemma}

\section{Irreducible $V_Q^\si$-modules}\label{irrep}

%
From now on, we assume $\si$ has order 2. The map $\si$ is extended by linearity to the complex vector space $\lieh=\CC\otimes_\ZZ Q$ and we denote by 
\begin{equation}\label{qbar1}
\pi_\pm=\frac{1}{2}(1\pm\si) \,, \qquad \al_\pm = \pi_\pm(\al)
\end{equation}
the projections onto the eigenspaces of $\si$.
Introduce the important sublattices
\begin{equation}\label{Lpm}
L_\pm=\lieh_\pm\cap Q,\, \qquad L=L_+\oplus L_-\subseteq Q\,,
\end{equation}
where $\lieh_\pm=\pi_\pm(\lieh)$.
Note that $\lieh=\lieh_+\oplus\lieh_-$ is an orthogonal direct sum.

In \cite{BE} we constructed and classified the irreducible $V_Q^\si$-modules and as a consequence of our result \cite[Theorem 3.7]{BE}, we realized all of them as submodules of twisted or untwisted $V_Q$-modules  \cite[Theorem 3.8]{BE}. 
For a positive-definite even lattice $Q$ and an order two automorphism $\sigma$ of $Q$, the first step in our construction is to restrict $Q$ to the sublattice $\bar{Q}$ satisfying $(\al|\sigma\al)\in2\ZZ$ for all $\al\in Q$. Under this condition, $\si^2=1$ on the lattice vertex algebra $V_{\bar{Q}}$, and $(V_Q)^{\si^2}=V_{\bar{Q}}$. Therefore
\begin{equation}\label{qbar2}
V_Q^\si = \bigl((V_Q)^{\si^2} \bigr)^\si = V_{\bar Q}^\si \,
\end{equation}
so that we may assume $|\si|=2$ on $V_Q$ and only work with the sublattice $\bar{Q}$. For simplicity, we use $Q$ instead of $\bar{Q}$ for the rest of this work. It is also shown in \cite{BE} that \,$V_Q^\si$ decomposes as a direct sum of irreducible \,$V_L^\si$-modules:
\begin{eqnarray}\label{orb}
V_Q^\si\cong\bigoplus_{\ga+L\in Q/L} V_{\ga_++L_+}\otimes V_{\ga_-+L_-}^{\eta(\ga)}\,.
\end{eqnarray}

\begin{theorem}[\cite{BE}]\label{result1}
Let\/ $Q$ be a positive-definite even lattice, and\/ $\sigma$ be an automorphism of\/ $Q$ of order two
such that\/ $(\al|\si\al)$ is even for all\/ $\al\in Q$. Then as a module over\/ $V_L^\si\cong V_{L_+}\otimes V_{L_-}^+$
each irreducible\/ $V_Q^\sigma$-module is isomorphic to one of the following$:$
\begin{align}
\label{AD1}
&\bigoplus_{\ga+L\in Q/L}V_{\ga_++\la+L_+}\otimes V_{\ga_-+\mu+L_-}
\qquad (2\mu\not\in L_-)\,, \\
\label{AD2}
&\bigoplus_{\ga+L\in Q/L}V_{\ga_++\la+L_+}\otimes V_{\ga_-+\mu+L_-}^{\epsilon\eta(\ga)}
\qquad (2\mu\in L_-)\,, 
\end{align}
where\/ $\la\in L_+^*$, $\mu\in L_-^*$, $\epsilon\in\{\pm\}$,  
\begin{align}
&\bigoplus_{\ga+L\in Q/L}V_{\ga_++\la+L_+}\otimes V_{L_-}^{T_\chi^{(\ga_-)},\epsilon_\ga}\qquad\epsilon_\ga=\epsilon\eta(\ga)c_\chi(\ga_-)
\,,\label{AD3}
\end{align}
where $\la\in (\pi_+Q)^*$, and $\chi$ is a central character of the group $\hat{L_-}/K$ associated to $L_-$ \upshape{(}c.f. \eqref{hatL}, \eqref{K}\upshape{)}.
Furthermore, every irreducible\/ $V_Q^\sigma$-module is a submodule of a\/ $V_Q$-module or a\/ $\si$-twisted\/ $V_Q$-module.
\end{theorem}

\begin{remark}\label{condition}
The proof is given in \cite{BE}. The only difference is the condition on $\la$ in \eqref{AD3}. Note that $(\pi_+Q)^*/\pi_+Q\subseteq L_+^*/L_+$ and set \begin{equation}\label{vgamma}
v_\ga=e^\ga+\eta(\ga)e^{\si\ga}=e^{\ga_+}\otimes (e^{\ga_-}+\eta(\ga)e^{-\ga_-}) \in V_{\ga_++L_+}\otimes V_{\ga_-+L_-}^{\eta(\ga)} \,.
\end{equation} 
Then by the fusion rules in \cite{ADL} we have 
\begin{align}
Y(v_\ga,z):e^\la\otimes V_{L_-}^{T_{\chi},\ep}\longrightarrow Y(e^{\ga_+},z)e^{\la}\otimes V_{L_-}^{T_{\chi}^{(\ga_-)},\ep_\ga},
\end{align}
where $ Y(e^{\ga_+},z)e^{\la}=z^{(\ga_+|\la)}\exp\Bigl( \sum_{n<0} \al_n \frac{z^{-n}}{-n} \Bigr)e^{\ga_++\la}$. To be an untwisted module implies that $(\ga_+|\la)\in\ZZ$ for all $\ga\in Q$, i.e., $\la\in(\pi_+Q)^*$.
Using a similar argument, we also have the restriction that $\la+\mu\in Q^*$ in \eqref{AD1} and \eqref{AD2}.
\end{remark}

In order to present the fusion products of irreducible $V_Q^\si$-modules, it is necessary to have more convenient notation.
We will denote the modules \eqref{AD1} by $U_{\la,\mu}$ and refer to them as modules of type 1. Similarly the modules \eqref{AD2} are denoted $U_{\la,\mu}^\pm$ and we refer to them as modules of type 2. Lastly, the modules \eqref{AD3} are denoted $T_{\la,\chi}^{\epsilon_0}$, where the sign $\epsilon_0$ is the eigenvalue of $\sigma$ on the term with $\ga=0$, and we refer to them as modules of twisted type.

\section{Quantum Dimension of Irreducible $V_Q^\si$-modules}\label{qd1}

The first important step is to reduce the quantum dimension over $V_Q^\si$ to the quantum dimension over the subalgebra $V_L^+$.
\begin{proposition}\label{qd}
We have
\begin{align}
q\dim_{V_L^+}V_Q^\si&=\left|Q/L\right|,\label{qdVL}
\end{align}
and the quantum dimension of any irreducible $V_Q^\si$-module $W$ can be determined from the quantum dimension of $W$ as a $V_L^+$-module. In particular,
\begin{align}
q\dim_{V_Q^\si}W&=\left|Q/L\right|^{-1}q\dim_{V_L^+}W.\label{qdW}
\end{align}
\end{proposition}

\begin{proof}
From the fusion rules of irreducible modules of $V_{L_+}$ and $V_{L_-}^+$, we have that $V_{\ga_++L_+}\otimes V_{\ga_-+L_-}^{\eta(\ga)}$ is a simple current for every \,$\ga\in Q$. The result \eqref{qdVL} then follows since $V_Q^\si$ has precisely $\left|Q/L\right|$ summands (cf. \eqref{orb}).
Now let $W$ be an irreducible $V_Q^\si$-module. Then 
\begin{align*}
q\dim_{V_Q^\si}W
&=\displaystyle\frac{q\dim_{V_L^+}W}{q\dim_{V_L^+}V_Q^\si}
=\left|Q/L\right|^{-1}q\dim_{V_L^+}W.
\end{align*}
\end{proof}
The quantum dimensions for untwisted type irreducible $V_Q^\si$-modules can now be determined using \eqref{qdW} and turn out to be independent of the lattice $Q$. 
\begin{theorem}\label{qdtype12}
Let \,$W$ be an irreducible $V_Q^\si$-module of untwisted type. Then we have the following quantum dimensions:
\begin{align}
q\dim_{V_Q^\si}W&=2,\hspace{8pt} \text{if $W$ is of type 1},\\
q\dim_{V_Q^\si}W&=1, \hspace{8pt}\text{if $W$ is of type 2.}\label{qdtype2}
\end{align}
In particular, the irreducible \,$V_Q^\si$-modules of type 2 are simple currents.
\end{theorem}

\begin{proof}
Let $W=\bigoplus_{\ga+L\in Q/L}V_{\ga_++\la+L_+}\otimes V_{\ga_-+\mu+L_-}$ be a $V_Q^\si$-module of type 1 (cf. \eqref{AD1}). Using \eqref{sum} and \eqref{tensor}, we have
\begin{align*}
q\dim_{V_{L}^+}W&=q\dim_{V_L^+}\bigoplus_{\ga+L\in Q/L}V_{\ga_++\la+L_+}\otimes V_{\ga_-+\mu+L_-}\\
&=\sum_{\ga+L\in Q/L}q\dim_{V_{L_+}\otimes V_{L_-}^+}V_{\ga_++\la+L_+}\otimes V_{\ga_-+\mu+L_-}\\
&=\sum_{\ga+L\in Q/L}q\dim_{V_{L_+}}V_{\ga_++\la+L_+}\cdot q\dim_{V_{L_-}^+}V_{\ga_-+\mu+L_-}.
\end{align*}
By the fusion rules of irreducible modules of $V_{L_+}$, the simple module \,$V_{\ga_++\la+L_+}$ is a simple current for $V_{L_+}$ so that \begin{align}
q\dim_{V_{L_+}}V_{\ga_++\la+L_+}=1.\label{qdtype1}
\end{align}
Similarly, $q\dim_{V_{L_-}}V_{\ga_-+\mu+L_-}=1$.
By the fusion rules of irreducible modules of $V_{L_-}^+$, both simple modules $V_{L_-}^\pm$ are simple currents for $V_{L_-}^+$ so that 
\begin{align}
q\dim_{V_{L_-}^+}V_{L_-}&=q\dim_{V_{L_-}^+}V_{L_-}^++q\dim_{V_{L_-}^+}V_{L_-}^-=2.
\end{align}
Since $V_{L_-}$ is a $V_{L_-}^+$-module, we can write
\begin{align}
q\dim_{V_{L_-}}V_{\ga_-+\mu+L_-}&=\frac{q\dim_{V_{L_-}^+}V_{\ga_-+\mu+L_-}}{q\dim_{V_{L_-}^+}V_{L_-}},
\end{align}
and therefore 
\begin{align}
q\dim_{V_{L_-}^+}V_{\ga_-+\mu+L_-}&=q\dim_{V_{L_-}}V_{\ga_-+\mu+L_-}\cdot q\dim_{V_{L_-}^+}V_{L_-}=2.
\end{align}
Hence $q\dim_{V_{L}^+}W=2|Q/L|$ and $q\dim_{V_Q^\si}W=2$ using \eqref{qdW}.

Now let $W=\bigoplus_{\ga+L\in Q/L}V_{\ga_++\la+L_+}\otimes V_{\ga_-+\mu+L_-}^{\epsilon_\ga}$ be a $V_Q^\si$-module of type 2, where each $\epsilon_\ga\in\{\pm\}$ (cf. \eqref{AD2}). By fusion rules of irreducible modules of $V_{L_-}^+$, both simple modules $V_{\ga_-+\mu+L_-}^\pm$ are simple currents. Then similarly by \eqref{sum}, \eqref{tensor}, 
\begin{align*}
q\dim_{V_{L}^+}W&=\sum_{\ga+L\in Q/L}q\dim_{V_{L_+}}V_{\ga_++\la+L_+}\cdot q\dim_{V_{L_-}^+}V_{\ga_-+\mu+L_-}^{\epsilon_\ga}\\
&=\sum_{\ga+L\in Q/L}1=|Q/L|.
\end{align*}
Hence $q\dim_{V_Q^\si}W=1$ by \eqref{qdW}.
\end{proof}

Quantum dimensions for irreducible $V_Q^\si$-modules of twisted type can also be realized in terms of quantum dimension of twisted type $V_{L_-}^+$-modules (cf. Proposition \ref{qd}). In order to write this relationship explicitly, we require a closer look at the fusion rules of twisted type modules in \cite{ADL}, from which we get the following fusion product:
\begin{align}\label{fp}
V_{L_-}^{T_\chi,\pm}\boxtimes(V_{L_-}^{T_\chi,\pm})'=\displaystyle\sum_{\substack{2\mu\notin L_-\\\chi^{(\mu)}=\chi}}V_{\mu+L_-}+\sum_{\substack{2\mu\in L_-\\\chi^{(\mu)}=\chi}}V_{\mu+L_-}^{\epsilon_\mu},
\end{align}
where each $\epsilon_\mu\in\{\pm\}$.
We first show that the quantum dimension corresponding to \,$\chi$\, and \, $\chi^{(\ga_-)}$\, are the same
for all $\ga\in Q/L$.
In order to use \eqref{fp}, we must investigate the solutions to the character equation \,$\chi^{(\mu)}=\chi$. 
The condition for the summations in \eqref{fp} for \,$\chi^{(\ga_-)}$\, is \,$\left(\chi^{(\ga_-)}\right)^{(\mu)}=\chi^{(\ga_-)}$, and this is equivalent to the character equation \,$\chi^{(\ga_-+\mu)}=\chi^{(\ga_-)}$:
\begin{align*}
\left(\chi^{(\ga_-)}\right)^{(\mu)}(U_\al)&=(-1)^{(\mu|\al)}\chi^{(\ga_-)}(U_\al)\\
&=(-1)^{(\mu+\ga_-|\al)}\chi(U_\al)\\
&=\chi^{(\ga_-+\mu)}(U_\al).
\end{align*}
It follows that there is a one-to-one correspondence between solutions to the character equations \,$\chi^{(\ga_-+\mu)}=\chi^{(\ga_-)}$\, and \,$\chi^{(\mu)}=\chi$. 
Using this correspondence, we then obtain:

\begin{lemma}\label{sameqd}
The quantum dimensions for the irreducible twisted type \,$V_{L_-}^+$-modules \,$V_{L_-}^{T_\chi^{(\ga_-)},\pm}$ are constant over \;$\ga\in Q/L$.
\end{lemma}


\begin{proposition}\label{qdQtoL}
Let \,$W$ be an irreducible $V_Q^\si$-module of twisted type. Then 
\begin{align}
q\dim_{V_Q^\si}W=q\dim_{V_{L_-}^+}V_{L_-}^{T_\chi,\pm}.
\end{align}
\end{proposition}

\begin{proof}
Let $W=\bigoplus_{\ga+L\in Q/L}V_{\ga_++\la+L_+}\otimes V_{L_-}^{T_\chi^{(\ga_-)},\epsilon_\ga}$ be a $V_Q^\si$-module of twisted type, where each $\epsilon_\ga\in\{\pm\}$ (cf. \eqref{AD3}).
Then using \eqref{sum}, \eqref{tensor}, \eqref{qdtype1}, \eqref{qdW}, and Lemma \ref{sameqd} we have
\begin{align*}
q\dim_{V_Q^\si}W&=\left|Q/L\right|^{-1}q\dim_{V_L^+}W\\
&=\left|Q/L\right|^{-1}\sum_{\ga+L\in Q/L}q\dim_{V_{L_+}}V_{\ga_++\la+L_+}\cdot q\dim_{V_{L_-}^+}V_{L_-}^{T_\chi^{(\ga_-)},\epsilon_\ga}\\
&=\left|Q/L\right|^{-1}\sum_{\ga+L\in Q/L} q\dim_{V_{L_-}^+}V_{L_-}^{T_\chi^{(\ga_-)},\epsilon_\ga}\\
&=q\dim_{V_{L_-}^+}V_{L_-}^{T_\chi,\epsilon_0}.
\end{align*}
\end{proof}

  Since $q\dim_{V_{L_-}^+}V_{L_-}^{T_\chi,\pm}=q\dim_{V_{L_-}^+}(V_{L_-}^{T_\chi,\pm})'$, the quantum dimensions of twisted type irreducible $V_Q^\si$-modules can be calculated using \eqref{fp}, Proposition \ref{qdQtoL}, and Theorem \ref{qdtype12}.
In view of this, we set 
\begin{align}\label{M}
 M&=
 L_-^*\cap\frac{1}{2}L_-,
\end{align}
so that the number of type 2 irreducible $V_{L_-}^\si$-modules is $|M/L_-|$.

It is clear for any central character $\chi$ of $\hat{L_-}/K$ that $\chi^{(\mu)}=\chi$ if and only if $\mu\in 2L_-^*$. 
It follows from \eqref{fp} and Proposition \ref{qdQtoL} that $q\dim_{V_Q^\si}W$ depends on 
those elements in $2L_-^*$ which correspond to distinct irreducible $V_{L_-}^\si$-modules of untwisted type.
We denote this finite subset of $2L_-^*$ by $R_\si$.
Then the number of summands in \eqref{fp} is $|R_\si|$ and the number of elements in $R_\si$ corresponding to type 2 irreducible $V_{L_-}^\si$-modules which appear in \eqref{fp} is $|R_\si\cap M|$. 
We can now write $q\dim_{V_Q^\si}W$ in terms of \,$R_\si$\, and \,$M$:
\begin{theorem}\label{qdtype3}
Let \,$W$ be an irreducible $V_Q^\si$-module of twisted type and \,$R_\si$, \,$M$\, be as defined above.
Then the square of the quantum dimension of \,$W$ over \,$V_Q^\si$ is given by
\begin{align}
(q\dim_{V_Q^\si}W)^2&=2|R_\si|-|R_\si\cap M|.\label{qdtw}
\end{align}
In particular, $(q\dim_{V_Q^\si}W)^2\in\ZZ$, and $(q\dim_{V_Q^\si}W)^2=|M|$ if $R_\si=M$.
\end{theorem}

\begin{proof}
By Proposition \ref{qdQtoL}, it is sufficient to find $q\dim_{V_{L_-}^+}V_{L_-}^{T_\chi,\pm}$. We do this using \eqref{fp} and counting the number of irreducible $V_{L_-}^+$-modules of untwisted type. 
There are precisely $|R_\si|-|R_\si\cap M|$ distinct irreducible modules of type 1 in the first sum in \eqref{fp} and $|R_\si\cap M|$ distinct irreducible modules of type 2 in the second sum in \eqref{fp}. The result now follows from \eqref{fp}, Proposition \ref{qdQtoL}, and Theorem \ref{qdtype12}.
\end{proof}

\begin{remark}
A special case of \eqref{qdtw} has been proved recently in \cite{DXY2} where the lattice $Q=L\oplus L$ is the direct sum of two copies of a positive definite even lattice $L$ of general rank and $\si$ is the transposition of the two summands (see also \cite{DXY1} for rank\,$L=1$). In this case they find \,$(q\dim_{V_Q^\si}W)^2=|L^*/L|$.
\end{remark}

\section{Fusion Products 
}\label{fp1}
Recall the notation \,$U_{\la,\mu}$, $U_{\la,\mu}^\pm$, and \,$T_{\la,\chi}^{\epsilon_0}$\, in Section \ref{irrep} for the irreducible $V_Q^\si$-modules. 
The following is the other main result about fusion products of irreducible $V_Q^\si$-modules:
\begin{theorem}
For any $\la\in L_+^*$ and $\mu\in L_-^*$, the following are fusion products for irreducible $V_Q^\si$-modules given in Theorem \ref{result1}:
\begin{align}
U_{\la,\mu}\;\boxtimes\;T_{\ga,\chi}^{\epsilon_0}&=T_{\la+\ga,\chi^{(\mu)}}^++T_{\la+\ga,\chi^{(\mu)}}^{-}\label{fr1}\\
U_{\la,\mu}^\epsilon\;\boxtimes\;T_{\ga,\chi}^{\epsilon_0}&=T_{\la+\ga,\chi^{(\mu)}}^{\epsilon_\mu},\quad \epsilon_\mu=\epsilon\epsilon_0c_\chi(\mu)\label{fr4}\\
U_{\la,\mu}\;\boxtimes\;U_{\la',\,\mu'}^{\epsilon}&=U_{\la+\la',\, \mu+\mu'}\label{fr2}\\
U_{\la,\mu}^\epsilon\;\boxtimes\;U_{\la',\mu'}^{\epsilon'}&=U_{\la+\la',\, \mu+\mu'}^{\epsilon\epsilon'}\label{fr3}\\
U_{\la,\mu}\;\boxtimes\;U_{\la',\,\mu'}&=
\begin{cases}
\sum_\epsilon (U_{\la+\la',\mu+\mu'}^\epsilon+U_{\la+\la',\mu-\mu'}^\epsilon)\\
U_{\la+\la',\mu+\mu'}+U_{\la+\la',\mu-\mu'}\\
U_{\la+\la',\mu+\mu'}+U_{\la+\la',\mu-\mu'}^++U_{\la+\la',\mu-\mu'}^-
\end{cases}\label{fr5}\\
T_{\la,\chi}^{\epsilon_0}\;\boxtimes\;T_{\la',\psi}^{\epsilon_0'}&=\sum_{\substack{\chi^{(\mu)}=\psi'\\2\mu\notin L_-}}U_{\la+\la',\mu}+\sum_{\substack{\chi^{(\mu)}=\psi'\\2\mu\in L_-}}U_{\la+\la',\mu}^{\epsilon_\mu}\label{fr6}
\end{align}
where \, $\epsilon_\mu=\epsilon_0\epsilon_0'c_\chi(\mu)(-1)^{2|\mu|^2}$ in \eqref{fr6}, the first case in \eqref{fr5} is when \,$2(\mu\pm\mu')\in L_-$\, the second case is when \,$2(\mu\pm\mu')\notin L_-$, and the third case is when 
\,$2(\mu+\mu')\notin L_-$\, and \,\,$2(\mu-\mu')\in L_-$.
\end{theorem}

\begin{proof}
We first prove \eqref{fr1}. By the fusion rules for irreducible \,$V_{L_+}$-modules in \cite{DL} and irreducible \,$V_{L_-}^\pm$-modules  in \cite{ADL}, the fusion rules of type $\displaystyle\binom{T_{\la+\ga,\chi^{(\mu)}}^{\pm}}{U_{\la,\mu}\;T_{\ga,\chi}^{\epsilon_0}}$ are nonzero and this yields the sum on the right hand side of \eqref{fr1}. Using \eqref{qdtw} and Theorem \ref{qdtype12} to count quantum dimensions of each side of \eqref{fr1}, we see that there can be no other nonzero fusion rules used in the fusion product. 

Similarly, the fusion rules in \cite{DL, ADL} yield nonzero intertwining operators of types $\displaystyle\binom{T_{\la+\ga,\chi^{(\mu)}}^{\epsilon_\mu}}{U_{\la,\mu}^\epsilon\;T_{\ga,\chi}^{\epsilon_0}}$, $\displaystyle\binom{U_{\la+\la',\, \mu+\mu'}}{U_{\la,\mu}^\epsilon\; U_{\la',\mu'}}$, $\displaystyle\binom{U_{\la+\la',\, \mu+\mu'}^{\epsilon\epsilon'}}{U_{\la,\mu}^\epsilon\; U_{\la',\mu'}^{\epsilon'}}$
and a similar argument as for \eqref{fr1} proves the fusion products \eqref{fr4}-\eqref{fr3}.
For \eqref{fr2}, note that $2\mu\notin L_-$\, and \,$2\mu'\in L_-$ implies $U_{\la+\la,\, \mu+\mu'}\cong U_{\la+\la,\, \mu-\mu'}$. 
Also, since $q\dim(U_{\la,\mu}^\epsilon\;\boxtimes\;U_{\la',\mu'}^{\epsilon'})=1$ by \eqref{qdtype2} and Proposition \ref{DJX}, the fusion product $U_{\la,\mu}^\epsilon\;\boxtimes\;U_{\la',\mu'}^{\epsilon'}$ is a simple current, and therefore must equal one irreducible module of type 2 (cf. \eqref{qdtw}, Theorem \ref{qdtype12}). We then obtain \eqref{fr3}, where the given rule of signs follows because $2\mu\in L_-$ implies $\pi_{\la,2\mu}=1$ (cf. \eqref{pi}). 

We now prove \eqref{fr5} and \eqref{fr6}. By the fusion rules in \cite{DL, ADL}, the intertwining operators of type $\displaystyle\binom{U_{\la+\la',\, \mu+\mu'}}{U_{\la,\mu}\; U_{\la',\mu'}}$ are nonzero only when $\mu+\mu'\notin 2L_-$ and the intertwining operators of type $\displaystyle\binom{U_{\la+\la',\, \mu+\mu'}^\pm}{U_{\la,\mu}\; U_{\la',\mu'}}$ are nonzero only when $\mu+\mu'\in 2L_-$. By counting quantum dimensions of both sides of \eqref{fr5} we see that there can be no other nonzero fusion rules used in the fusion product for either case. This proves \eqref{fr5}.
Again by fusion rules in \cite{DL, ADL}, the intertwining operators of type $\displaystyle\binom{U_{\la+\la',\mu}}{T_{\la,\chi}^{\epsilon_0}\;T_{\la',\psi}^{\epsilon_0'} }$ are nonzero only when $\chi^{(\mu)}=\psi'$ and $2\mu\notin L_-$, and the intertwining operators of type $\displaystyle\binom{U_{\la+\la',\mu}^{\epsilon_\mu}}{T_{\la,\chi}^{\epsilon_0}\;T_{\la',\psi}^{\epsilon_0'} }$ are nonzero only when $\chi^{(\mu)}=\psi'$ and $ 2\mu\in L_-$ (cf. \eqref{chiprime}). 
In order to count the quantum dimension of the right-hand side of \eqref{fr6}, we count the number of solutions $\mu+L_-$ to the character equation $\chi^{(\mu)}=\psi'$. Let \,$R_\si^{\chi,\psi}$\, be the solution set, where each element is distinct under $\si$. Clearly then for any $\mu+L_-\in R_\si^{\chi,\psi}$, we have \,$\mu+L_-+R_\si\subset R_\si^{\chi,\psi}$ (cf.  \eqref{R}). We show these two sets are equal, i.e., that $\mu, \mu'\in R_\si^{\chi,\psi}$ implies $\mu-\mu'+L_-\in R_\si$:
\begin{align*}
\chi^{(\mu-\mu')}(U_\al)&=(-1)^{(\mu-\mu'|\al)}\chi(U_\al)\\
&=(-1)^{(\mu|\al)+(\mu'|\al)+(\al|\al)}\chi(U_\al)\\
&=(-1)^{(\mu|\al)+\frac{1}{2}(\al|\al)}(-1)^{(\mu'|\al)+\frac{1}{2}(\al|\al)}\chi(U_\al)\\
&=(-1)^{(\mu|\al)+\frac{1}{2}(\al|\al)}\psi(U_\al)\\
&=\chi(U_\al),
\end{align*}
where the last two steps use $\chi^{(\mu)}=\psi'$. Therefore $|R_\si^{\chi,\psi}|=|R_\si|$.
By counting quantum dimensions of both sides of \eqref{fr6} we now see there can be no other nonzero fusion rules used in the fusion product. Finally, the rule for the sign $\epsilon_\mu$ in \eqref{fr6} follows from the fusion rules in Theorem \ref{ADL} and \eqref{V'}. This proves \eqref{fr6}.
\end{proof}

\section{Examples}\label{ex}
Here we show explicit details of quantum dimensions and fusion products of twisted type orbifold modules in the cases of a 2-permuatation orbifold
and for the $A_n$ root lattice with a Dynkin diagram automorphism.

\subsection{$(V_K\otimes V_K)^{\ZZ_2}$ orbifolds
}\label{sirr}
The case when $Q$ is an orthogonal direct sum of two copies of an arbitrary rank positive definite even lattice $K$ with the 2-cycle permutation has recently been studied in \cite{DXY2}. Here we present how some of our results can simplify in these cases and show the explicit results in the special case of the $A_2$ root lattice. In this section, we use the symbol $K$  to avoid confusion with \eqref{Lpm}.


Let \,$Q=K\oplus K$\; be an orthogonal direct sum, where $K=\langle\al_1,\ldots,\al_n\rangle$ is a positive definite even lattice with rank\;$n$. Set $(\al_i|\al_i)=2k_i$
 and for $\ga\in K$, set 
 $\ga^1=(\ga,0)$\, and\, $\ga^2=(0,\ga).$
Then a 2-cocycle $\ep$ on $K$ can be extended to a 2-cocycle on $Q$ for which we may take $\eta=1$ on $Q$ (cf. \eqref{twlat4}).
Let $\si$ be the 2-cycle permutation $\si(a,b)=(b,a)$. 
Then the vectors 
\begin{align}
\al^i&=\al_i^1+\al_i^2=(\al_i,\al_i),\\ 
\be^i&=\al_i^1-\al_i^2=(\al_i,-\al_i), \label{beta}
\end{align}
for $i=1,\ldots,n$, are eigenvectors of $\si$ with eigenvalues $\pm1$
and
\begin{align}
(\al^i|\al^i)=&\,4k_i=(\be^i|\be^i),\label{4k}\\
(\al^i|\al^j)=2&(\al_i|\al_j)=(\be^i|\be^j).
\end{align}
These relations then imply $Q=\bar{Q}$ (cf. Section 3).

First we describe the orbifold $V_Q^\si$. We have (cf. \eqref{Lpm}, \eqref{orb})
\begin{align}
L_+&=\langle\al^1,\ldots,\al^n\rangle,\quad
L_-=\langle\be^1,\ldots,\be^n\rangle, \\
V_Q^\si
&=\bigoplus_{(b_1,\ldots,b_n)}\left(V_{\frac{1}{2}\sum b_i\al^i+L_+}\otimes V_{\frac{1}{2}\sum b_i\be^i+L_-}^+\right),
\end{align}
where each $b_i\in\{0,1\}$ and there are $|Q/L|=
2^n$ summands. 
We have 
\begin{align}
M
=&\left\langle\frac{\be^i}{2}\,\bigg|\,i=0,\ldots,n\right\rangle,\\
R_\si\cap M=&\;\text{span}_{\ZZ_2}\left\{\displaystyle \frac{\be^i}{2}\,\bigg|\,(\al_i|K)\in2\ZZ\right\},
\end{align}
and in this case, the Gram matrix for $K$ can determine the size of $R_\si\cap M$.
\begin{lemma}\label{RM1}
Let $M$ be as in \eqref{M}. Then $|R_\si\cap M|=1$ if and only if for every $i$, $(\al_i|\al_j)$ is odd for at least one $j$.
\end{lemma}
\begin{proof}
Let $\frac{\be^i}{2}\in M$ and $j$ be such that $(\al_i|\al_j)$ is odd. Then
\begin{align}
\left(\frac{\be^i}{2}\bigg|\be^j\right)&=(\al_i|\al_j),\qquad
\chi^{(\frac{\be^i}{2})}(e^{\be^j})=-\chi(e^{\be^j}),
\end{align}
so that $\frac{\be^i}{2}\notin R_\si$.
\end{proof}

Due to \eqref{4k}, we have 
\begin{align}
\chi'=\chi
\end{align}
 for all central characters $\chi$ of $L_-$. Hence all irreducible $V_Q^\si$-modules of twisted type are self dual.

\subsubsection{Rank 1 case}\label{A2D}
Here we provide explicit details of Theorem \ref{qdtype3} in the case when $K$ is a rank 1 even positive definite lattice. These results will correspond to the results previously done in \cite{DXY1}. 

Write $K=\ZZ\ga$, where $(\ga|\ga)=2k$ for some positive integer $k$ and let $\si$ be the 2-cycle permutation. Then it is easy to see that 
\begin{align}
L_-=\ZZ\be,&\qquad(\be|\be)=4k,\\
L_-^*=\displaystyle\frac{1}{4k}L_-,&\qquad 2L_-^*=\displaystyle\frac{1}{2k}L_-.
\end{align}
We also have $|R_\si\cap M|=2$ and the other elements of $L_-^*/L_-$ are paired under $\si$. Therefore the quantum dimension of any twisted type irreducible $V_Q^\si$-module $W$ is given by Theorem \ref{qdtype3}:
\begin{align}
(q\dim_{V_Q^\si}W)^2&=2|R_\si|-|R_\si\cap M|\\
&=2\left(\frac{2k-2}{2}+2\right)-2\\&=2k.
\end{align}

\subsubsection{$A_2\oplus A_2$ case}\label{A2D}
Here we provide explicit details of Theorems \ref{result1} 
and \ref{qdtype3} in the case when $K=\ZZ\al_1+\ZZ\al_2$ is the $A_2$ root lattice. 

Consider the orthogonal direct sum \,$Q=K\oplus K$ 
and let $\si$ be the 2-cycle permutation.
Then 
\begin{align}
L_+&=\langle\al^1,\al^2\rangle,\quad
L_-=\langle\be^1,\be^2\rangle,\\
|\al^i|^2=4&=|\be^i|^2,\quad
(\al^1|\al^2)=-2=(\be^1|\be^2),
\end{align}
and the orbifold is given by
\begin{align}
V_Q^\si&=\bigoplus_{(b_1, b_2)}\left(V_{\frac{1}{2}\sum b_i\al^i+L_+}\otimes V_{\frac{1}{2}\sum b_i\be^i+L_-}^+\right),
\end{align}
where $b_1,b_2\in\{0,1\}$. 
To describe the dual lattices \,$L_\pm^*$\, we set
\begin{align}
\de^1&=\frac{\al^1+2\al^2}{6},\quad \de^2=\frac{2\al^1+\al^2}{6},\\
\rho^1&=\frac{\be^1+2\be^2}{6},\quad \rho^2=\frac{2\be^1+\be^2}{6},
\end{align}
so that $L_+^*=\langle\de^1,\de^2\rangle$\, and \,$L_-^*=\langle\rho^1,\rho^2\rangle$
with quotients
$L_\pm^*/L_\pm\simeq\ZZ_2\times\ZZ_6.$
Note that 
\,$\si(\rho^i+L_-)=5\rho^i+L_-$,\;
$2\rho^2+L_-=4\rho^1+L_-$,
\,and\, $2\de^2+L_+=4\de^1+L_+$.
By Theorem \ref{result1}, the type 1 irreducible $V_Q^\si$-modules are
\begin{align}\label{type1A2}
U_{\la_{ij}, \,\mu}=\bigoplus_{(b_1, b_2)}\left(V_{\la_{ij}+\frac{1}{2}\sum b_i\al^i+L_+}\otimes V_{\mu+\frac{1}{2}\sum b_i\be^i+L_-}\right),
\end{align}
where $\la_{ij}=i\de^1+j\de^2$, $0\leq i\leq5$, $j=0, 1$ and $\mu\in\{\rho^1, 2\rho^1, \rho^2, \rho^1+5\rho^2\}$.
Set \,$\mu=k\rho^1+l\rho^2$.
Then for  \eqref{type1A2} to be untwisted, we also require that \,$i+k\in2\ZZ$ and $j+l\in2\ZZ$ (cf. Remark \ref{condition}).
This gives 12 distinct irreducible modules of type 1.

Now we find \,
\begin{equation}
M/L_-=\left\langle3\rho^1+L_-,3\rho^2+L_-\right\rangle\simeq\ZZ_2^2,
\end{equation}
(cf. \eqref{M}) so that the type 2 irreducible $V_Q^\si$-modules are given by $U_{\la_{ij}, \,\mu}^\pm$, with $\mu+L_-\in M/L_-$. In order for $U_{\la_{ij}, \,\mu}^\pm$ to be untwisted we also require that \,$i+k\in2\ZZ$ and $j+l\in2\ZZ$ (cf. Remark \ref{condition}).
This gives 24 distinct irreducible modules of type 2.

For the orbifold modules of twisted type, consider the following labelling of the central characters:
\begin{center}
\begin{tabular}{c|rrrc}
&1&$e^{\be^1}$&$e^{\be^2}$&$e^{\be^1+\be^2}$\\
\hline
$\chi_{00}$&1&1&1&1\\
$\chi_{10}$&1&1&$-1$&$-1$\\
$\chi_{01}$&1&$-1$&1&$-1$\\
$\chi_{11}$&1&$-1$&$-1$&$1$
\end{tabular}
\end{center}
For $\ga=c_1\beta^1+c_2\beta^2$, notice that $\chi_{ij}(e^\ga)=(-1)^{c_1j+c_2i}$. We can then use this labelling to show the following properties:
\begin{lemma}
Let $b_i, c_i\in\{0,1\}$ and $\mu=\frac{1}{2}b_1\beta^1+\frac{1}{2}b_2\beta^2$. 
Then we have
\begin{align}
c_{\chi}\left(\mu\right)&=\chi(e^{b_1\be^1})\chi(e^{b_2\be^2}),\label{cchiA_2A_2}\\
\chi_{ij}^{(\mu)}&=\chi_{i+b_1,j+b_2}\label{chiij}\\
\chi^{(r\rho^1+k\rho^2)}(e^{\vec{c}\cdot\vec{\be}})&=(-1)^{rc_2+kc_1}\chi(e^{\vec{c}\cdot\vec{\be}}),
\end{align}
where $\vec{c}\cdot\vec{\be}=c_1\beta^1+c_2\beta^2$ and the indices in \eqref{chiij} are taken modulo $2$.
\end{lemma}

Now since $(\pi_+ Q)^*/(\pi_+ Q)=\langle2\de^1+\pi_+ Q\rangle\simeq\ZZ_3$,
there are $3\cdot4=12$ orbifold modules of twisted type:
\begin{align}
T_{2i\de^1,\chi_{mn}}^{\epsilon_0}=&\bigoplus_{\mu=(b_1,b_2)}V_{2i\de^1+\frac{1}{2}\sum b_i\al^i+L_+}\otimes V_{L_-}^{T_{\chi_{m+b_1,n+b_2}},\epsilon_\mu},
\end{align}
where $i=0,1, 2$, $\epsilon_\mu=\epsilon c_\chi(\mu)$
(cf. \eqref{cchiA_2A_2}).

Now we determine the quantum dimensions for irreducible modules of twisted type:
\begin{align}
R_\si=\{0, 2\rho^1\},&\quad | R_\si\cap M|=1,\\
q\dim_{V_Q^\si}T_{\la,\chi_{ij}}^{\epsilon_0}&=\sqrt{3},\label{qdimtype3}
\end{align}
(cf. Lemma \ref{RM1}) and a few of their fusion products:
\begin{align}
T_{\la,\chi_{00}}^{\epsilon_0}\;\boxtimes\;T_{\la',\chi_{10}}^{\epsilon_0'}&=U_{\la+\la',\rho^1}
+U_{\la+\la',3\rho^1}^{\epsilon_0\epsilon_0'},\\
T_{\la,\chi_{00}}^{\epsilon_0}\;\boxtimes\;T_{\la',\chi_{01}}^{\epsilon_0'}&=U_{\la+\la',\rho^2}
+U_{\la+\la',3\rho^2}^{\epsilon_0\epsilon_0'},\\
T_{\la,\chi_{00}}^{\epsilon_0}\;\boxtimes\;T_{\la',\chi_{11}}^{\epsilon_0'}&=U_{\la+\la',\rho^1+5\rho^2}
+U_{\la+\la',\rho^1+\rho^2}^{\epsilon_0\epsilon_0'}.
\end{align}
The other fusion products  are similar.
Note that each fusion product can be written in the form \,$U_{\la,\mu} + U_{\la,\mu'}^\pm$\, ($\la\in L_+, \;\mu,\mu'\in L_-$),
and that
$q\dim(T_{\la,\chi_{ij}}^{\epsilon_0}\;\boxtimes\;T_{\la',\chi_{kl}}^{\epsilon_0'})=3$
for\, $i,j,k,l\in\{0,1\}$ (cf. \eqref{qdimtype3}).


\subsection{The $A_n$ root lattice.}\label{An}
Consider the $A_n$ Dynkin diagram with the simple roots $\{\alpha_1,\ldots, \alpha_n\}$, 
%
%
the associated root lattice $Q=\sum_{i=1}^n\ZZ\al_i$, 
and the Dynkin diagram automorphism $\si:\alpha_i \leftrightarrow \alpha_{n-i+1}$.
Set 
\begin{equation}
\alpha^i=\alpha_i+\alpha_{n-i+1}\quad \text{and}\quad \beta^i=\alpha_i-\alpha_{n-i+1}.
\end{equation}
The construction of all irreducible $V_Q^\si$-modules was done in \cite{E}. Since the results depend on the parity of $n$, we treat the cases separately.
\subsubsection{Results for $n$ odd}
Throughout this section, set $l=\frac{n-1}{2}$ and let $i<l+1$. Then
\begin{equation*}
|\al^i|^2=4=|\be^i|^2,\quad(\alpha^i |\beta^j)=0,\quad(\alpha^i |\alpha^{i+1})=-2=(\beta^i |\beta^{i+1}),
\end{equation*}
and $(\alpha^i |\alpha^j)=0=(\beta^i |\beta^j)$ otherwise, and
\begin{eqnarray*}
L_+&=&\sum_{i=1}^l\mathbb{Z}\alpha^i+\mathbb{Z}\alpha_{l+1},\quad
L_-=\sum_{i=1}^l\mathbb{Z}\beta^i.
\end{eqnarray*}
Therefore $Q=\bar{Q}$ (cf. Section 3), and the orbifold is given by
\begin{equation}\label{orbAnodd}
V_Q^\si\simeq\displaystyle\bigoplus_{\left(b_1,\ldots,b_l\right)}\left(V_{\frac{1}{2}\sum b_i\alpha^i+L_+}\otimes V_{\frac{1}{2}\sum b_i\beta^i+L_-}^+\right),
\end{equation}
where $b_1,\ldots,b_l\in\{0,1\}$.
Now the Gram matrix for $L_-$ is twice the Gram matrix for $A_l$ so that we may use the dual lattice $\{\la_1,\ldots,\la_l\}$ for $A_l$ to describe the dual lattice for $L_-$. A basis for $L_-^*$ is then
\begin{align*}
\mu_i =\frac{1}{2(l+1)}((l-i+&1)\be^1+2(l-i+1)\be^2\\&+i(l-i+1)\be^i+i(l-i)\be^{i+1}+\cdots+i\be^l),
\end{align*}
where $i=1,\ldots,l$.
Then there are $l+4$ irreducible $V_Q^\si$-modules of untwisted type (see \cite{E} Theorem 5.5.5):
\begin{align*}
U_{0}^\pm=&\displaystyle\bigoplus_{(b_1,\ldots,b_l)}\left(V_{\frac{1}{2}\sum b_i\alpha^i+L_+}\otimes V_{\frac{1}{2}\sum b_i\be^i+L_-}^\pm\right),\\
\hspace{0.1in}
U_{
k}=&\displaystyle\bigoplus_{(b_1,\ldots,b_l)}\left(V_{\frac{1}{2}\epsilon_k\alpha_{l+1}+\frac{1}{2}\sum b_i\alpha^i+L_+}\otimes V_{k\mu_l+\frac{1}{2}\sum b_i\be^i+L_-}\right),\quad k=1,\ldots,l\\
U_{
l+1}^\pm=&\displaystyle\bigoplus_{(b_1,\ldots,b_l)}\left(V_{\frac{1}{2}\epsilon_l\alpha_{l+1}+\frac{1}{2}\sum b_i\alpha^i+L_+}\otimes V_{(l+1)\mu_l+\frac{1}{2}\sum b_i\be^i+L_-}^\pm\right),
\end{align*}
where $b_1,\ldots,b_l\in\{0,1\}$ and $\epsilon_i=0$ for $i$ odd and $\epsilon_i=1$ for $i$ even.
From this it is clear that 
$|R_\si|=\left\lfloor\frac{l}{2}\right\rfloor+2$ \,and\, $|R_\si\cap M|=2.$
If $W$ is an irreducible $V_Q^\si$-module of twisted type, then from Theorem \ref{qdtype3} we have 
\begin{eqnarray*}
(q\dim_{V_Q^\si}W)^2=
\begin{cases}
l+2,&l\text{ even},\\
l+1,&l\text{ odd}
\end{cases}.
\end{eqnarray*}

\subsubsection{Results for $n$ even}
Throughout this section, set $l=\frac{n}{2}$ and let $i<l+1$. Then for $i<l$,
\begin{align*}
|\al^i|^2&=4=|\be^i|^2,\quad(\alpha^i |\beta^j)=0,\quad(\alpha^i |\alpha^{i+1})=-2=(\beta^i |\beta^{i+1}),
\end{align*}
$(\beta^l |\beta^l)=6$, $(\al^l |\al^l)=2$, and $(\alpha^i |\alpha^j)=0=(\beta^i |\beta^j)$ otherwise. Then 
\begin{eqnarray*}
\bar{Q}
=\sum_{i=1}^{l-1}\mathbb{Z}\alpha_i+\ZZ\al^l+\ZZ\be^l+\sum_{i=l+2}^n\mathbb{Z}\alpha_i,
\end{eqnarray*}
and the cosets $\bar{Q}/L$ are in correspondence with $\{0,1\}$-valued $(l-1)$-tuples (cf. Section 3).
Therefore the orbifold is given by
\begin{equation}\label{orbAnodd}
V_Q^\si\simeq\displaystyle\bigoplus_{\left(b_1,\ldots,b_{l-1}\right)}\left(V_{\frac{1}{2}\sum b_i\alpha^i+L_+}\otimes V_{\frac{1}{2}\sum b_i\beta^i+L_-}^+\right),
\end{equation}
where $b_1,\ldots,b_{l-1}\in\{0,1\}$.
A basis of $L_-^*$ was computed in  \cite[Proposition 5.6.2]{E}:
\begin{align*}
\mu_i = \frac{1}{2(2l+1)}((2l-2i+1)\beta^1&+\cdots+i(2l-2i+1)\beta^i\\&+i(2l-2i-1)\beta^{i+1}+\cdots+i\beta^l),
\end{align*}
where $1\leq i\leq l$.
Set  $\gamma = \frac{1}{2}\left(\beta^1+\beta^3+\cdots+\beta^s\right)$, where $s=  \begin{cases} 
   l, &  l \hspace{0.05in}\text{odd} \\
   l-1,& l \hspace{0.05in}\text{even}\\
       \end{cases}$.
       
Then there are $l+4$ irreducible $V_Q^\si$-modules of untwisted type (see \cite{E} Theorem 5.6.6):
\begin{align*}
U_0^\pm=&\displaystyle\bigoplus_{(b_1,\ldots,b_{l-1})}\left(V_{\frac{1}{2}\sum b_i\alpha^i+L_+}\otimes V_{\frac{1}{2}\sum b_i\be^i+{L_-}}^\pm\right)\\
U_k=&\displaystyle\bigoplus_{(b_1,\ldots,b_{l-1})}\left(V_{\frac{1}{2}\sum b_i\alpha^i+L_+}\otimes V_{2k\mu_l+\frac{1}{2}\sum b_i\be^i+{L_-}}\right),\hspace{0.1in} k=1,\ldots,l-1,\\
U_\ga^\pm=&\displaystyle\bigoplus_{(b_1,\ldots,b_{l-1})}\left(V_{\frac{1}{2}\sum b_i\alpha^i+L_+}\otimes V_{\ga+\frac{1}{2}\sum b_i\be^i+{L_-}}^\pm\right),
\end{align*}
where $b_1,\ldots,b_l\in\{0,1\}$.
Note that \,$\frac{\be^l}{2}\notin R_\si$. From this it is clear that 
\begin{eqnarray*}
|R_\si|=
\begin{cases}
l+1,&l\text{ even},\\
l,&l\text{ odd}
\end{cases}
,\quad 
|R_\si\cap M|=
\begin{cases}
2,&l\text{ even},\\
1,&l\text{ odd}
\end{cases}.
\end{eqnarray*}
If $W$ is an irreducible $V_Q^\si$-module of twisted type, then from Theorem \ref{qdtype3} we have 
\begin{eqnarray*}
(q\dim_{V_Q^\si}W)^2=
\begin{cases}
2l,&l\text{ even},\\
2l-1,&l\text{ odd}
\end{cases}.
\end{eqnarray*}

\section*{Acknowledgements}
The author is grateful to Bojko Bakalov for helpful discussions and revisions of drafts. The author is also grateful to the reviewer for thoughtful comments and suggestions.

\bibliographystyle{amsalpha}

\end{document}